\documentclass[11pt]{article}  
\usepackage{fullpage} 
\usepackage{amsthm}
\usepackage{amsmath}
\usepackage{amssymb}
\usepackage{todonotes}
\usepackage{mathtools}
\usepackage{xspace}
\usepackage{xcolor}
\usepackage{verbatim}
\usepackage{enumitem}
\usepackage{multirow}
\usepackage{nicefrac}
\usepackage{hyperref}
\usepackage{cite}

\newcommand{\MM}{\mbox{\sc MaxMatching}}

\newcommand{\disj}{\mbox{\sc Disj}}
\newcommand{\perm}{\mbox{\sc Perm}}
\newcommand{\prob}{\mbox{\rm P}}

\newcommand{\Diam}{\mbox{\sc Diameter}\xspace}
\newcommand{\Degeneracy}{\mbox{\sc Degeneracy}}
\newcommand{\MaxClique}{\mbox{\sc MaxClique}}
\newcommand{\MaxIS}{\mbox{\sc MaxIS}}
\newcommand{\MinCover}{\mbox{\sc MinVC}}

\newcommand{\kCol}{\mbox{\sc Coloring}}
\newcommand{\Col}{\mbox{\sc Coloring}}

\newcommand{\CCN}{\mathsf{CC}^\mathsf{N}}
\newcommand{\NP}{\mathrm{NP}}
\newcommand{\MA}{\mathrm{MA}}
\newcommand{\cgst}{\textsc{congest}\xspace}

\newtheorem{theorem}{Theorem}
\newtheorem{lemma}{Lemma}

\newtheorem{corollary}{Corollary}
\newtheorem{fact}{Fact}

\title{Semi-Streaming Algorithms for Graph Property Certification}

\author{
Avinandan Das\thanks{Email: {\tt adas33745@gmail.com}. Additional support from the ANR projects DUCAT (ANR-20-CE48-0006) and QuData (ANR-18-CE47-0010) and ERC Consolidator Grant Distributed Biological Algorithms. }  \\
{\small Department of Computer Science}\\
{\small Aalto University} \\
{\small Finland}
\and 
Pierre Fraigniaud\thanks{Additional support from the ANR projects DUCAT (ANR-20-CE48-0006), QuData (ANR-18-CE47-0010), and ENEDISC (ANR-24-CE48-7768-01). } \\
{\small IRIF}\\
{\small CNRS and Universit\'e Paris Cit\'e} \\
{\small France}
\and 
Ami Paz \\
{\small LISN}\\
{\small CNRS and Universit\'e Paris-Saclay} \\
{\small France}
\and 
Adi Ros\'en\thanks{Email: {\tt adiro@irif.fr}. Research partially supported by ANR project PREDICTIONS.}\\
{\small IRIF}\\
{\small CNRS and Universit\'e Paris Cit\'e} \\
{\small France}
}

\date{}

\begin{document}
%%%%%%%%%%%%%%%%%%%%%%%%%%%%%%%%%%%%%%%%%%%%%%%%%%%%%%%%%%%%%%%%%
\maketitle

\begin{abstract}
We introduce the  {\em certification} of solutions to graph problems when access to the input is restricted. This topic has received a lot of attention in the distributed computing setting, and we introduce it here in the context of \emph{streaming} algorithms, where the input is too large to be stored in memory. 

Given a graph property~$\prob$, a \emph{streaming certification scheme} for $\prob$ is a \emph{prover-verifier} pair where the prover is a computationally unlimited but non-trustable oracle, and the verifier is a streaming algorithm. For any input graph, the prover provides the verifier with a \emph{certificate}. The verifier then receives the input graph as a stream of edges in an adversarial order, and must check whether the certificate is indeed a \emph{proof} that the input graph satisfies~$\prob$. The main complexity measure for a streaming certification scheme is its \emph{space complexity}, defined as the sum of the size of the certificate provided by the oracle, and of the memory space required by the verifier.

We give streaming certification schemes for several graph properties, including maximum matching, diameter, degeneracy, and coloring, with space complexity matching the requirement of \emph{semi-streaming}, i.e., with space complexity $O(n\,\mbox{polylog}\, n)$ for $n$-node graphs. All these problems do {\em not} admit semi-streaming algorithms, showing that also in the (semi) streaming setting, certification is sometimes easier than calculation (like $NP$). For each of these properties, we provide upper and lower bounds on the space complexity of the corresponding certification schemes, many being tight up to logarithmic multiplicative factors. We also show that some graph properties are hard for streaming certification, in the sense that they cannot be certified in semi-streaming, as they require $\Omega(n^2)$-bit certificates. 
\end{abstract}

\newpage

%%%%%%%%%%%%%%%%%%%%%%%%%%%%%%%%%%%%%%%%%%%%%%%%%%%%%%%%%%%%%%%%%
\section{Introduction}
\label{sec:intro}
%%%%%%%%%%%%%%%%%%%%%%%%%%%%%%%%%%%%%%%%%%%%%%%%%%%%%%%%%%%%%%%%%

In this paper, we introduce a new framework for the {\em certification} of solutions to computing problems, when access to the input is restricted. More specifically, we are interested in the case where the input is too large to be stored in memory, and is accessed as a \emph{stream} of items by a limited-memory algorithm, usually referred to in the literature as a  {\em streaming algorithm} (cf.~\cite{AlonMS99,McGregor14}). Recall that streaming algorithms are algorithms for processing data streams in which the input is presented as a sequence of items, and are designed to operate with limited memory (which prevents them from storing the whole input in memory).
For our study of streaming certification, we concentrate on graph problems and the semi-streaming regime~\cite{FeigenbaumKMSZ05,McGregor14}, i.e., the streaming algorithm is limited to memory of size $O(n~\mbox{polylog~}n)$, for an $n$-node graph. An $n$-node input graph $G$ with $m$ edges is thus presented to the streaming algorithm as a sequence of edges, in arbitrary, possibly adversarial, order, where each edge is merely a pair $(i,j)\in [n]\times [n]$. Note that if $|E(G)|\gg n~\mbox{polylog~}n$ then the streaming algorithm cannot keep all edges in memory for computing its output. 

We study decision problems, i.e., whether an input graph has a certain property such as containing a matching of size at least~$k$, or being $3$-colorable. 
We study properties that cannot be decided by a semi-streaming algorithm, and examine the question whether these properties can be {\em certified} by a semi-streaming algorithm. 
More precisely, our setting consists of a prover-verifier pair, where the semi-streaming algorithm (the verifier) receives a {\em certificate} from an untrusted oracle (the prover) before the stream of items starts. The certificate is supposed to be a proof 
that the graph $G$ represented by the stream of edges to be given satisfies a given property~$\prob$, e.g., ``$G$~does not have an independent set of size at least~$k$''. 
The certification scheme must satisfy two properties for being correct. 

\begin{description}
    \item[Completeness:] On a legal instance, i.e., for a graph $G$ satisfying the property~$\prob$, the prover must be able to provide the verifier with a certificate such that, for every order of edges in the stream, the verifier algorithm outputs {\tt accept}. 
    \item[Soundness:] On an illegal instance, i.e., for a graph $G$ not satisfying the property~$\prob$, it must be the case that, for every possible certificate provided by the prover, and for any order of the edges in the stream, the verifier outputs {\tt reject}. 
\end{description}

Note that, similar to the definition of, e.g., NP, the verifier is not subject to any requirement regarding ``wrong'' certificate given for a legal instance. 
Completeness only requires the verifier to accept a legal instance whenever it is given a ``right'' certificate for this instance. On the other hand, soundness requires that, for an illegal instance, the prover cannot fool the verifier with a certificate that would lead it to accept.  

We investigate graph decision problems that cannot be decided by a semi-streaming algorithm, i.e., by a semi-streaming algorithm without a certificate. For these problems, we give upper and lower bounds on the memory size of the verifier, and on the size of the certificate provided by the prover. 

We stress the fact that, in addition to its conceptual interest, 
the framework of \emph{streaming certification scheme} may find applications in, e.g., cloud computing, when an expensive computation on a huge input might be delegated to, say, a cloud (or any external powerful computing entity). A user may then want to check that the output provided by the cloud is correct. To this end, the cloud could be asked to provide the user with a (small) certificate. While the user cannot store the entire input in memory, it can receive the input as a stream (possibly in adversarial order), and it should be able to verify the correctness of the cloud's output using the certificate.

%-----------------------------------------------------------------
\subsection{Our Results} 
%-----------------------------------------------------------------

We focus on graph parameters, including maximum matching, degeneracy, diameter, chromatic number, vertex-cover,  maximum independent set, and maximum clique. 
For each of these parameters, we aim at certifying whether it is at most, or at least a given threshold value~$k$ (which may be fixed, or given as input). For instance, $\MM_{\leq k}$ (resp., $\MM_{\geq k}$) is the boolean predicate corresponding to whether the input graph given as a stream has a maximum matching of size at most~$k$ (resp., at least~$k$). 

Table~\ref{tab:our-results} summarizes our results. 
For each graph property under consideration, this table provides the space complexity of a certification scheme for it, that is, the size $c$ of the certificate provided by the prover, and the space-complexity $m$ of the deterministic verifier (not counting the certificate size, which is accessed by the verifier as a read-only memory). 
The lower bounds for the certification schemes refer to the sum of the prover's certificate size and the verifier's memory space, i.e., $c+m$. 
We also provide lower bounds on the  space complexity~$s$ of streaming algorithms \emph{deciding} the properties under scrutiny (without  certificate). 

\begin{table}[tb]
    \small
    \centering
    \renewcommand{\arraystretch}{1.5}
    \setlength{\tabcolsep}{3.5pt}
    \begin{tabular}{|c||c||c|c|c||c||c|c|c|}
        \hline
        \multirow{2}{*}{Property} 
        & \multicolumn{4}{|c||}{$\leq k$} & \multicolumn{4}{|c|}{$\geq k$} \\
        & $s$ & $c$           & $m$           & $c+m$     & $s$ & $c$               & $m$       & $c+m$ \\
\hline
\MM     
        &   $\Theta(n^2)$ \cite{FeigenbaumKMSZ05}     & $O(n)$        & $O(n\log n)$  &     ?     & $\Theta(n^2)$ \cite{FeigenbaumKMSZ05}     & $O(k\log n)$  & $O(\log n)$    & $\Omega(n)$ \\
        &      &               &               &          &    &    $O(n\log \Delta)$  & $O(n)$ &  \\
\hline
\Degeneracy 
        & $\Theta(n^2)$ &  $O(n\log n)$  & $O(n\log k)$  & $\Omega(n)$ &  $\Theta(n^2)$ & $O(k\log n)$  & $O(n\log k)$ & ? \\
\hline
\Diam 
        &  $\Theta(n^2)$ & - &  -             &      $\Theta(n^2)$     & $\Theta(n^2)$ & $O(n\log k)$      & $O(\log n)$ &  $\Omega(n)$ \\
\hline
\kCol   
        & $\Theta(n^2)$ & $O(n\log k)$ & $O(\log n)$ &  $\Omega(n\log n)$   &  $\Theta(n^2)$         & - & - &  $\Theta(n^2)$ \\
\hline
\MaxClique & $\Theta(n^2)$ &    -    &      -         &    $\Theta(n^2)$       & $\Theta(n^2)$ & $O(k\log n)$  &      $O(\log n)$         & ? \\
\hline
\MaxIS  &  $\Theta(n^2)$ &    -      &     -          &     $\Theta(n^2)$      & $\Theta(n^2)$ & $O(k\log n)$  &   $O(\log n)$  & ?\\
\hline
\MinCover & $\Theta(n^2)$ &  $O(k\log n)$  &  $O(\log n)$ &     ?      & $\Theta(n^2)$ &       -            &      -         & $\Theta(n^2)$ \\
\hline
    \end{tabular}    
    \caption{\sl Summary of our results regarding the space complexity of deciding vs.\ certifying graph properties in streaming. Notations: $s$ denotes the space complexity of deciding the property without certificates; $c$ and $m$ are the parameters of the certification scheme, where $c$ denotes the certificate size, and $m$ the space complexity of the verifier.}
    \label{tab:our-results}
\end{table}

\sloppy
As an example, it is known~\cite{FeigenbaumKMSZ05} that deciding whether the input graph has a perfect matching requires a streaming algorithm to have space complexity $\Omega(n^2)$. But, we show that both $\MM_{\leq k}$ and $\MM_{\geq k}$ can be \emph{certified} in semi-streaming. 
Specifically, there is a certification scheme for $\MM_{\leq k}$ with certificate size $O(n)$ bits and verifier's space complexity $O(n\log n)$, and there is a certification scheme for $\MM_{\geq k}$ with certificate size $O(k\log n)$ bits and verifier's space complexity $O(\log n)$. 
We give another certification scheme for $\MM_{\geq k}$ with certificate size $O(n\log \Delta)$ bits and verifier's space complexity $O(n)$, where $\Delta$ denotes the maximum degree of the input graph. These latter certification schemes for $\MM_{\geq k}$ are essentially the best that one can hope for, up to poly-logarithmic factors, as we also show that any certification scheme for $\MM_{\geq k}$ with $c$-bit certificates and verifier's space complexity $m$ bits satisfies $c+m=\Omega(n)$ bits. 

All the other entries in the tables read similarly. Note that $\MM$ and $\Degeneracy$ are two problems that behave very well with respect to semi-streaming certification, in the sense that, for each of them, the two decision problems relative to the question whether the solution is at most~$k$ or at least~$k$ can be certified in semi-streaming. Instead, all the other graph problems listed in Table~\ref{tab:our-results} have one of their two variants ($\leq k$ or $\geq k$) that cannot be certified in semi-streaming, as $c+m=\Omega(n^2)$ for this variant. As we shall see further in the paper, the $\Omega(n^2)$ lower bound on $c+m$ is even quite strong for some problems, in the sense that it holds even for constant~$k$. 

The cells with a question mark ``?'' correspond to open problems, that is, problems for which we do not have non-trivial answers. For instance, we do not know whether it is possible to certify $\MM_{\leq k}$ with $c+m=o(n)$. 
The cells marked with ``-'' correspond to problems for which the trivial upper bound $O(n^2)$ on $c+m$ is tight (up to constant multiplicative factors). 

%-----------------------------------------------------------------
\subsection{Related Work}
\label{sec:related-work}
%-----------------------------------------------------------------

\paragraph{Streaming and semi-streaming.}
Our work studies streaming algorithm, aiming at handling inputs of huge size that cannot be stored in memory, and are usually received by the (streaming) algorithm as a stream of items, many times in adversarial order. 
This line of work finds its origins in the seminal work of Alon et al.~\cite{AlonMS99}. Specifically, we are interested in graph problems, where one usually seeks streaming algorithms with memory size of $O(n~\mbox{polylog}~n)$, for $n$ nodes graphs, a regime usually referred to as {\em semi streaming}, as coined by the work of Feigenbaum et al.~\cite{FeigenbaumKMSZ05}. This regime allows one to store the nodes of the graph, but not all edges (unless the graph is very sparse). Many works consider this regime, e.g.,~\cite{EmekR16,ChenKL22,Kapralov13,PazS19,FeigenbaumKMSZ08,AhnG13} to name a few.

\paragraph{Distributed certification.} 

Our work is very much inspired by the theory of \emph{distributed certification}~\cite{FKP13,GS16,KKP10}. In this setting, the processing units of a distributed system are connected as the $n$-nodes of a labeled graph $G=(V,E)$, where the label of a node models the state of that node. The goal is to certify that the global state of the system encoded by the collection of labels is legal with respect to some boolean predicate on labeled graphs. 
A typical example is certifying that a collection of pointers (each node pointing to a neighbor) forms a spanning tree, or even a minimum-weight spanning tree~\cite{KK07}. A \emph{distributed certification scheme} for a boolean predicate $\prob$ on labeled graphs is a pair prover-verifier. The prover is a non-trustable oracle assigning a certificate to each node, and the verifier is a distributed algorithm whose objective is to check that the collection of certificates is a proof that the labeled graph satisfies $\prob$. If yes, then all nodes must accept, otherwise at least one node must reject. The verifier is ideally performing in a single round of communication between neighboring nodes. The main complexity measure is the size of the certificates assigned to the nodes. We refer to~\cite{FF16} for a survey. Interestingly, this concept has been extended to distributed interactive proofs~\cite{NaorPY20}, to distributed zero-knowledge proofs~\cite{BickKO22}, and to distributed quantum proof~\cite{FraigniaudGNP21}.

\paragraph{Annotated Streaming Model.}

A related model of streaming verification, a.k.a the \emph{annotated streaming model}, was formulated in~\cite{ChakrabartiCMT14}. Distributed certification can be viewed as a specific instantiation of annotated streaming. To see why, recall that the annotated streaming model consists of a (computationally powerful, and deterministic) \emph{oracle} and a (possibly randomized) \emph{algorithm}. The oracle ``annotates'' each item of the input stream to help the algorithm. The aim for the  algorithm is to compute some function $f(x)$ on input stream $x = (x_1,\ldots, x_m)$ where each $x_i$ belongs to some universe $\mathcal{U}$ (e.g., for a graph $G=(V,E)$, $\mathcal{U}=E$). More formally, an \emph{annotated data streaming scheme} consists of a pair $(\mathfrak{h},\mathcal{A})$ where $\mathfrak{h}= (h_1,h_2,\ldots)$ is the \emph{annotation}, and $\mathcal{A}$ is the \emph{algorithm}, that is, $\mathcal{A}$ is a (possibly randomized) streaming algorithm. Their aim is to compute a function $f$ on the input stream $x$. The algorithm $\mathcal{A}$ either outputs a value taken from the range of the function~$f$, or it outputs~$\bot$.  Two types of schemes are considered, which depend on whether the oracle knows the input stream in advance (i.e., before the stream begins).	
\begin{itemize}
		\item Prescient scheme: The oracle knows the whole input stream $x$ in advance. In this case, each item $x_i$ is augmented with a bit string which depends on~$x$. More formally,  for every $i\in [m]$, $h_i\colon \mathcal{U}^m\rightarrow \{0,1\}^*$ is a function, and the streaming algorithm $\mathcal{A}$ receives the annotated stream $\big( (x_1,h_1(x)), \ldots,(x_m,h_m(x))\big )$ as input. 
        
		\item Online scheme: The oracle does not know the input stream in advance, which is revealed to the oracle in the order of the stream. In this case, each item $x_i$ is augmented with a bit string which depends only on the stream seen so far. More formally, for every $i\in [m]$, $h_i\colon \mathcal{U}^i\rightarrow \{0,1\}^*$ is a function,  and the streaming algorithm $\mathcal{A}$ receives the annotated stream $\big ( (x_1,h_1(x_1)),(x_2,h_2(x_1,x_2)),\ldots,(x_m,h_m(x_1,\ldots,x_m))\big )$ as input.
\end{itemize}

A (prescient or online) scheme $(\mathfrak{h},\mathcal{A})$ for a function~$f$ in the annotated streaming model, where $\mathfrak{h}= (h_1,h_2,\ldots)$, must satisfies the following completeness and soundness properties. 

\begin{description}
     \item[Completeness:] For every input stream $x=(x_1,\dots,x_m)$,  
     $$\Pr[\mbox{$\mathcal{A}$ outputs $f(x)$ on input $\big((x_1,h_1),\dots,(x_m,h_m)\big)$ }]\geq \nicefrac23$$ 
     
     \item[Soundness:] For every input stream $x=(x_1,\dots,x_m)$, and for every annotation $\mathfrak{h}' = (h_1',\ldots, h_m')$,   $$\Pr[\mbox{$\mathcal{A}$ outputs $f(x)$ or $\bot$ on input $\big((x_1,h'_1),\dots,(x_m,h'_m)\big)$]}\geq \nicefrac23$$ 
 \end{description}
 
In other words, the completeness and soundness of the annotation scheme $(\mathfrak{h},\mathcal{A})$ imply that 
(1)~given the \emph{correct} annotation~$\mathfrak{h}$, with probability at least~$\nicefrac 23$ the algorithm $\mathcal{A}$ correctly computes the function~$f$, and 
(2)~for any annotation $\mathfrak{h}'$ (presumably \emph{adversarial}, whose purpose is to make $\mathcal{A}$ output a wrong value for $f(x)$), with probability at least~$\nicefrac 23$ the algorithm $\mathcal{A}$ either outputs the correct value for $f(x)$, or outputs~$\bot$, where the latter can be interpreted as  the fact that the algorithm has detected that it has not received the correct annotation for the input~$x$. 

A great deal of work has been done on the annotated streaming model. For instance, the classical statistical problems like selection, frequency moments, and frequency based functions were considered in~\cite{ChakrabartiCGT14,ChakrabartiCMT14,Ghosh20}. Graph problems have also been studied extensively in this model, including maximum matching, minimum weight spanning tree, and triangle counting~\cite{ChakrabartiCMT14,CormodeMT13,Thaler16,ChakrabartiG19,ChakrabartiGT20}. 
A more general model of \emph{streaming interactive proof} has been considered as well in~\cite{ChakrabartiC0TV19, CormodeTY11}, and, very recently, \cite{CormodeDGH24}~has initiated the study of zero-knowledge proofs in interactive streaming.

A detailed comparison between annotated streaming and streaming certification is provided in Section~\ref{sec:annotation-vs-certification}.

\paragraph{Communication Complexity.}

Our lower bounds for semi-streaming certification schemes are obtained by reduction from communication complexity, where an efficient streaming certification scheme implies an efficient non-deterministic communication two-party protocol for problems such as set-disjointness ($\disj$). Several prior streaming lower bounds were established using reductions to communication complexity (see, e.g., \cite{VerbinY11,DBLP:conf/icalp/BeraCG20,AbboudCKP21}), but they were not assuming certificates provided by an oracle.
Reductions from communication complexity are also common in proving lower bounds for distributed graph algorithms in the so-called \cgst model \cite{AbboudCKP21,DBLP:Feuilloley,SarmaHKKNPPW12,DBLP:CzumajK20,GrossmanKP20,HolzerW12}. In this setting, a few
lower bounds for non-deterministic variants (e.g., proof labeling schemes, distributed interactive proofs, etc.) were obtained by using reductions to non-deterministic communication complexity~\cite{Censor-HillelPP20,GS16}, but many graph problems were considered in~\cite{AbboudCKP21}, which provides a unified framework for proving lower bounds applying to both streaming algorithms, and distributed algorithms, by reduction to communication complexity.
We build upon their framework, and upon some of their reductions in several cases (for diameter, minimum vertex cover, maximum independent set, and coloring), with the non-determinism constraint imposed by streaming certification.

As far as specific problems are concerned, proper $(d+1)$-coloring, where $d$ denotes the degeneracy of the input graph, was considered in the semi-streaming model~\cite{DBLP:conf/icalp/BeraCG20}.
While we do not consider this problem, their lower bound construction inspired our lower bound for certifying degeneracy.
Computing degeneracy in the semi streaming model was also considered in the past~\cite{Farach-ColtonT14,Farach-ColtonT16}, but we are not aware of any lower bounds for the problem.
Clique detection in the \cgst model was considered in~\cite{DBLP:CzumajK20}, and our lower bound uses a similar gadget.
Similarly, our coloring certification lower bound uses a gadget which was  presented earlier for a similar problem in the \cgst model~\cite{DBLP:Feuilloley},
and our lower bound for diameter uses a gadget already used in the \cgst model~\cite{HolzerW12}.
Finally, our upper bound for maximum matching uses a dual problem of finding a small set of nodes whose removal leaves many odd-sized connected components. 
This idea, known as the Tutte-Berge formula, found applications also in other certification settings, such as
annotated streams~\cite{Thaler16}, and proof labeling schemes~\cite{Censor-HillelPP20}.

%%%%%%%%%%%%%%%%%%%%%%%%%%%%%%%%%%%%%%%%%%%%%%%%%%%%%%%%%%%%%%%%%
\section{Model and Preliminary Results}
%%%%%%%%%%%%%%%%%%%%%%%%%%%%%%%%%%%%%%%%%%%%%%%%%%%%%%%%%%%%%%%%%

The streaming  model consists of a memory-bounded (space-constrained) algorithm that receives its input as a stream of items, processes them in their order of arrival, and, after the stream ends, answer a pre-defined question on the input.

%-----------------------------------------------------------------
\subsection{Streaming Certification Schemes}
%-----------------------------------------------------------------

In this work we focus on decision problems, i.e., on streaming algorithms that decide, for a given language, if the input it receives is in that language. For such problems we define the notion of \emph{streaming certification scheme}. A streaming certification scheme  has two components: a \emph{prover}, which is a computationally unlimited but non-trustable oracle, and a \emph{verifier}, which is a \emph{deterministic} streaming algorithm. The streaming algorithm needs not run in polynomial time, although for our upper bounds it would be desirable that it does. For a fixed decision problem~$\prob$, e.g., whether a given graph is 3-colorable, the prover acts as follows. Given an input $x$ for $\prob$, the prover provides the verifier with a \emph{certificate} $c\in\{0,1\}^*$. This certificate depends on the input~$x$, but not on the order of the items in the stream  corresponding to~$x$ (e.g., it does not depend on the order in which the edges of an input graph $G$ are streamed). The certificates is designed by the prover for the verifier, and is stored in a read-only memory that can be accessed by the latter. 
The input~$x$ is given to the verifier as a stream of items that are processed by the verifier, which can also access the certificate~$c$ at will. When the stream ends, the verifier has to return its output, i.e.,  \emph{accept} or \emph{reject}. (We consider solely the insertion-only model, i.e., no items are ever deleted once they appear in the stream.)
The pair prover-verifier is a correct (streaming) certification scheme for~$\prob$ if the following two conditions hold for every possible input $x$ for~$\prob$. 

\begin{description}
    \item[Completeness:] If $x \models \prob$ then there exists a certificate $c\in\{0,1\}^\star$ such that, for every order in which the items of $x$ are given as a stream~$s$, the verifier accepts $(s,c)$.
    \item[Soundness:] If $x\not\models \prob$ then, for every certificate $c\in\{0,1\}^\star$, and for every order in which the items of $x$ are given as a stream~$s$, the verifier rejects $(s,c)$.
\end{description}

For measuring the quality of a (correct) certification scheme (i.e., satisfying both completeness and soundness), we focus on the following complexity measure: 
the size of the certificates provided by the prover, plus the space-complexity of the verifier. 
That is, given a decision problem~$\prob$, for every $n\geq 1$, we are interested in minimizing the following quantity over all prover-verifier pairs:
\[
\max_{\substack{x \models\prob \\ |x|=n}}\; 
\min_{c\, \in \, \mbox{\footnotesize accept}(x)} \;
\max_{\pi  \in\Sigma_n  } \;
(|c|+\mbox{space}(\pi(x),c))
\]
where $\Sigma_n$ denotes the set of permutation over $[n]$, $\pi(x)$ is the stream of $n$ items for $x$ ordered according to~$\pi$,  $\mbox{accept}(x)$ is the set of certificates leading the verifier to accept $x$ for every order $\pi$ in which $x$ can be streamed as input, i.e., 
\[
\mbox{accept}(x)=\{c\in \{0,1\}^\star \mid \forall \pi \in \Sigma_n, \; \mbox{verifier}(\pi(x),c)=\mbox{accept}\},
\]
 and $\mbox{space}(\pi(x),c)$ is the space-complexity (i.e., the memory requirement) of the verifier on the input stream~$\pi(x)$ and certificate~$c$. For instance, if $\prob$ is a graph problem then, for every graph $G=(V,E)$, $\pi(G)$ is a mere ordering of the edges of~$G$. 

%--------------------------------------------------
\subsection{Annotation vs.\ Certification}
\label{sec:annotation-vs-certification}
%--------------------------------------------------

Streaming certification can be viewed as a variant of annotated streaming (see Section~\ref{sec:related-work}) \emph{focusing on decision functions}, i.e., functions~$f$ defined as 
\[
f(x)=\left\{\begin{array}{ll}
1 & \mbox{if $x\models \prob$} \\
0 & \mbox{otherwise.}
\end{array}\right.
\]
Moreover, in the context of streaming certification, we deal with prescient schemes only. Furthermore, the whole certificate $c(x)$ is given at once, right before, or together with the fist item $x_1$ of the stream $x=(x_1,\dots,x_m)$. That is, $c(x)=h_1(x)$, and $h_i(x)=\varnothing$ for all $i>1$. A slight difference here is that the certificate depends only of the input (that is of the graph~$G$), and not on the order in which the items (i.e., the edges) are streamed. Finally, we consider the semi-streaming model, and we do not insist on polylogarithmic space complexity of the algorithm. 
However, completeness and soundness conditions of streaming certification are strengthened compared to annotated streaming: They must hold with probability~1, i.e., we focus on \emph{deterministic} verifier algorithms~$\mathcal{A}$. In other words, one does not need to amplify the outcome of the verifier for being convinced of the correctness of the answer (e.g., by performing several passes of the stream). If $f(x)=1$ then the certificate provided by the prover must lead the verifier to systematically accept, and if $f(x)=0$ then, whatever certificate is provided by the prover, the verifier must  systematically  reject. 

Overall, for decision functions, streaming certification schemes can be thought of as an equivalent of~$\NP$ but for streaming. Instead, annotation schemes can be thought of as an equivalent of $\MA\cap \text{co-}\MA$ for streaming\footnote{Recall the MA stands for Merlin-Arthur, i.e., the verifier is randomized, and is allowed to err with small probability.}. Indeed, not only annotation schemes require that, for an input $x$ such that $f(x)=1$ (i.e., $x\models \prob$), the oracle provides the (randomized) verifier with an annotation leading it to output~1 (i.e., \emph{accept}), but also, for an input $x$ such that $f(x)=0$ (i.e., $x\not\models \prob$), the oracle must provide the (randomized) verifier with an annotation leading it to output~0 (i.e., \emph{reject}). It is solely when the annotation is not the one expected by the verifier for the given input stream that it is allowed to output~$\bot$. Nevertheless, we show that the weaker ``$\NP$-like'' streaming certification schemes are sufficient to certify many graph properties in the semi-streaming model.

%-----------------------------------------------------------------
\subsection{Graph Properties and Semi-Streaming}
%-----------------------------------------------------------------

In the framework of deciding graph property, the input is an $n$-node graph given as a stream of edges. 
The set $[n]=\{1,\ldots,n\}$ of nodes is known in advance, and each edge is given as a pair of nodes, i.e., a pair of integers in $[n]\times [n]$.
We are primarily interested in streaming algorithms that have linear space, up to polylogarithmic factors,
which is usually referred to as \emph{semi-streaming}. 
That is, the verifier has an upper bound on memory of $O(n\; \mbox{polylog} \; n)$ bits for $n$-node graphs, allowing the algorithm to store a poly-logarithmic number of edges (whereas the number of edges may be quadratic).
%all the nodes, but not necessarily all the edges. 
Let us exemplify the notion of certification in this context.  

\paragraph{Example.} 

Let $\prob$ be the decision problem asking whether a given graph $G$ is $3$-colorable,
and let us consider the following possible certification scheme for~$\prob$.
For every $3$-colorable $n$-node graph~$G$, the prover provides the algorithm with a certificate $c\in\{0,1,2\}^n$ where, for every $i\in [n]=\{1,\dots,n\}$, the $i$-th entry $c[i]$ of~$c$ is a color assigned to the $i$-th node of~$G$ in a proper 3-coloring of~$G$. The edges of $G$ are then provided to the verifier as a stream, in an arbitrary order.
Upon reception of an edge $\{i,j\}\in [n]\times [n]$, the verifier merely checks that $c[i]\neq c[j]$. 
If this is the case for all edges, then the verifier accepts, and otherwise it rejects.
This elementary scheme is complete, as for every $3$-colorable graph the prover can provide a proper $3$-coloring as a certificate, and the verifier will accept the graph with this certificate.
On the other hand, for any non 3-colorable graph~$G$, every coloring~$c$ of the nodes with colors picked in $\{0,1,2\}$ yields the existence of at least one monochromatic edge $e$.
Hence, for any possible certificate provided by a verifier, the algorithm will reject upon the arrival of $e$ (if it did not already rejected earlier), 
rendering the scheme sound.
The space complexity of this certification scheme is $O(n)$ as $c$ can be encoded on $2n$ bits, and checking whether $c[i]\neq c[j]$ for any given $i$ and $j$ consumes $O(\log n)$ bits of memory. 
As we shall see later in the text, the space complexity of this straightforward certification scheme is actually essentially optimal.

\paragraph{The Problems.}

Given a $n$-node graph $G$ and a positive integer $k$, we consider the following  problems:
\begin{itemize}
    \item \MM. Decide whether every maximum matching of $G$ has size~$k$.
    \item \Degeneracy. Decide whether $G$ is $k$-degenerate.
    \item \Diam. Decide whether diameter of $G$ is~$k$.
    \item \kCol. Decide whether the chromatic number of $G$ is~$k$.
    \item  \MinCover,\ \MaxIS\ and \MaxClique. Decide whether $G$ has, respectively, minimum vertex cover, maximum independent set, and maximum clique of size~$k$.
\end{itemize}
In fact, we shall consider variants of these problems, depending on whether one questions equality to~$k$, or at least~$k$, or at most~$k$. 

\paragraph{Notation.}

For a graph parameter~$\rho$ (e.g., chromatic number, diameter, dominating number, etc.), let $\textsc{Rho}$ be the corresponding decision problem, i.e., the problem that takes as input a graph~$G$ and an integer~$k\geq 0$ that we refer to as the \emph{threshold}, and asks whether $\rho(G)=k$. We distinguish the case where the threshold $k$ is given as input from the case where it is fixed. 
\begin{itemize}
\item Input threshold: In this framework, we shall consider the decision problems $\textsc{Rho}_{=}$, $\textsc{Rho}_{\leq}$, and $\textsc{Rho}_{\geq}$, which take as input a graph~$G$ and a threshold~$k\in\mathbb{N}$, and respectively ask whether $\rho(G)= k$, $\rho(G)\leq k$ or $\rho(G)\geq k$. In the context of streaming, the threshold $k$ is the first item received by the algorithm, before it receives the first edge of the input graph~$G$. 

\item Fixed threshold: In this framework, for any given integer $k\geq 1$, we shall consider the decision problems $\textsc{Rho}_{= k}$, $\textsc{Rho}_{\leq k}$ and $\textsc{Rho}_{\geq k}$, which take as input a graph~$G$, and respectively ask whether $\rho(G)= k$, $\rho(G)\leq k$, or $\rho(G)\geq k$. That is, the threshold $k$ is fixed, and the algorithm might be designed for this specific threshold. 
\end{itemize}

\medskip

The certificate size as well as the memory requirement of the algorithm will be given as a function of~$k$ and~$n$, and the big-$O$ and big-$\Omega$ notations hide absolute constants, i.e., independent of~$k$ and~$n$. Some of our lower bounds hold for constant $k$ (i.e., for fixed parameter), while other cases are only proved for $k$ which grows with~$n$ (i.e., for input parameter). The upper bounds, on the other hand, are always given as a function of both $n$ and $k$. The following is a straightforward, but important remark regarding streaming algorithms (without certificates).

\begin{lemma}\label{lemma:equals_k_to_atleast_k}
	Let $\rho$ be a graph parameter, and let $k\geq 0$ be an integer. If there exists a streaming algorithm for deciding $\rho(G)\leq k$, and a streaming algorithm for deciding $\rho(G)\leq k-1$, both with space complexity $s(n)$ bits on $n$-node graphs~$G$, then there exists a streaming algorithm for $\rho(G)= k$ with space complexity $O(s(n))$ bits. The same holds under the assumption that there exist streaming algorithms for deciding $\rho(G)\geq k$ and for deciding $\rho(G)\geq k+1$, both with space complexity $s(n)$ bits.
\end{lemma}

\begin{proof}
The results follow from the simple observation that 
\[
(\rho(G)=k) \; \equiv \; (\rho(G)\leq k) \land \neg(\rho(G)\leq k-1), 
\]
and, similarly,  $(\rho(G)=k) \; \equiv \; (\rho(G)\geq k) \land \neg(\rho(G)\geq k+1)$.     
\end{proof}

We complete this section by establishing a series of space lower bounds motivating the use of streaming certification schemes for several central graph theoretical problems.  

\begin{theorem}
\label{theo:lower-bounds-decision}
    All the following problems require $\Omega(n^2)$ bits of memory to be decided by a (possibly randomized) streaming algorithm: 
    $\MM_{\leq}$ and  $\MM_{\geq}$, 
    $\Degeneracy_{\leq}$ and  $\Degeneracy_{\geq}$, 
    $\MinCover_{\leq}$ and $\MinCover_{\geq}$, 
    $\MaxIS_{\leq}$ and $\MaxIS_{\geq}$,
    $\MaxClique_{\leq}$ and $\MaxClique_{\geq}$, and $\kCol_{\geq}$.
    Furthermore, $\Diam_{\leq 2}$, $\Diam_{\geq 3}$, and
    $\kCol_{\leq 3}$ also require $\Omega(n^2)$ bits of memory to be decided by a (possibly randomized) streaming algorithm.
\end{theorem}

\begin{proof}
Thanks to Lemma~\ref{lemma:equals_k_to_atleast_k}, it is sufficient to show the claimed lower bounds for equality to the threshold~$k$. 
For example, it has been shown in~\cite{FeigenbaumKMSZ05} that any (possibly randomized) streaming algorithm deciding whether the $n$-node input graph $G$ has a perfect matching (i.e., a matching of size $k=n/2$) has space complexity $\Omega(n^2)$ bits. Therefore, $\MM_{=}$ requires $\Omega(n^2)$ space, and  thus  $\MM_{\leq}$ and $\MM_{\geq}$ as well. 

Similarly, it has been shown in~\cite{DBLP:conf/icalp/BeraCG20} that distinguishing $n$-node graphs of degeneracy~$k$ from those with chromatic number $k+2$ (which have degeneracy at least $k+1$) requires $\Omega(n^2)$ space in the randomized streaming model. Therefore, $\Degeneracy_{=}$ requires $\Omega(n^2)$ space.

Finally, it has been proved in~\cite{AbboudCKP21} that any (possibly randomized) streaming algorithms for  $\MinCover_{=}$, $\MaxIS_{=}$, or $\MaxClique_{=}$ require $\Omega(n^2)$ space.

We now move on with the case where the threshold $k$ is fixed (and not part of the input). 
The lower bound graphs from~\cite{HolzerW12} to prove lower bounds for \Diam in distributed settings also imply a $\Omega(n^2)$ lower bound for $\Diam_{=2}$ in streaming (see Lemma~\ref{lem:family-of-streaming-lower-bound} and Theorem~\ref{theo:diam_lb} for more details). This implies that $\Diam_{\leq 2}$ requires $\Omega(n^2)$ space since $\Diam_{\leq 1}$ merely requires $O(\log n)$ bits for counting the edges to check that the graph is a clique (of $n(n-1)/2$ edges). Similarly, $\Diam_{\geq 3}$ requires $\Omega(n^2)$ space because $\Diam_{\geq 2}$ is solvable by an algorithm with $O(\log n)$ space.  

For coloring, it has been proved in~\cite{AbboudCKP21} that any (possibly randomized) streaming algorithms for $\Col_{=3}$ require $\Omega(n^2)$ space. Therefore, $\Col_{\leq 3}$ requires $\Omega(n^2)$ space as $\Col_{\leq 2}$ (i.e., bipartiteness) can be decided with $O(n\log n)$ space by maintaining a spanning forest~\cite{McGregor14}. Thanks to Lemma~\ref{lemma:equals_k_to_atleast_k}, at least one of the two decision problems $\Col_{\geq 3}$ or $\Col_{\geq 4}$ also requires $\Omega(n^2)$ space, and so does $\Col_{\geq}$. 
\end{proof}

%%%%%%%%%%%%%%%%%%%%%%%%%%%%%%%%%%%%%%%%%%%%%%%%%%%%%%%%%%%%%%%%%
\section{Semi-Streaming Certification}
%%%%%%%%%%%%%%%%%%%%%%%%%%%%%%%%%%%%%%%%%%%%%%%%%%%%%%%%%%%%%%%%%

In this section, we consider central graph parameters
that are of specific interest in the context of certification schemes as the corresponding decision problems cannot be solved in the  semi-streaming framework (cf. Theorem~\ref{theo:lower-bounds-decision}).
That is, we are now investigating whether there is a way to \emph{certify} solutions for these problems in a semi-streaming manner. 
Interestingly, as we shall see in this section, for maximum matching and degeneracy, both directions (i.e., ``at most'' and ``at least'') can be certified in semi-streaming. 
For the other considered graph parameters, such a semi-streaming certification exists solely for one direction.  

%----------------------------------------------------------------
\subsection{Maximum Matching}
%---------------------------------------------------------------- 

For maximum matching, we show the following. 

\begin{theorem}\label{theo:MM:upper-bound}
    There exists a semi-streaming certification schemes for $\MM$. Specifically: 
    \begin{itemize}
        \sloppy
        \item There exist a certification scheme for $\MM_{\geq}$  with certificate size $O(k\log n)$ bits, and a verifier's space complexity of $O(\log n)$ bits for all $n$-node graphs and threshold $k \geq 1$.
        \item There exist a certification scheme for $\MM_{\geq}$ with certificate size $O(n\log \Delta)$ bits, and verifier's space complexity of $O(n)$ bits for all $n$-node graphs of  maximum degree~$\Delta$.
        \item There exists a certification scheme for $\MM_{\leq}$ with certificate size $O(n)$ bits, and verifier's space complexity is $O(n\log n)$ bits  for all $n$-node graphs.
    \end{itemize}
\end{theorem}

The proof is split in two parts. 
We first show that certifying the existence of a matching of size at least~$k$ can be done in semi-streaming,
with certificate size $O(n\log \Delta)$ bits or $O(k\log n)$ bits (cf. lemmas~\ref{lem:MM-at-least-k} and~\ref{lem:another-MM-at-least-k}).
Next, we show that certifying the non-existence of a matching of size larger than~$k$  can also be done in semi-streaming (cf. Lemma~\ref{lem:MM-at-most-k}). 
We start with the simpler, perhaps obvious, scheme.

\begin{lemma}\label{lem:MM-at-least-k}
     There exists a semi-streaming certification scheme for $\MM_{\geq}$. Specifically, the certificate size of the scheme is $O(k\log n)$ bits, and the space complexity of the algorithm is $O(\log n)$ bits.      
\end{lemma}

\begin{proof}
Let $G=(V,E)$ be a graph for which the  maximum matching has size at least~$k$. Let $M$ be a (maximum) matching of $G$ of size~$k$. 
The prover gives $M$ as certificate to the verifier, as a set of $k$ edges, encoded using $O(k\log n)$ bits. Given the certificate $M$, the verifier checks whether (1)~$M$ is a matching and has size~$k$, and (2)~$M\subseteq E$. Condition~(1) is using $M$ only (i.e., not the stream), and checking it requires counting the number of edges in $M$ and exhaustively going over $M$ for each $\{u,v\}\in M$ to check if $u$ and $v$ appear exactly once in $M$, which requires $O(\log n)$ bits. Checking condition~(2) requires counting the number of edges   $\{x,y\}$,  appearing in the stream, are included  in the certificate  $M$. If there are $k$ such streamed edges, then the verifier concludes that $M\subseteq E$. This again requires only $O(\log n)$ bits. If conditions~(1) and~(2) are satisfied, then the verifier accepts, otherwise it rejects. 

Completeness is satisfied by construction. 
For establishing soundness, let $G$ be a graph with maximum matching size smaller than~$k$. For every set $M$ of $k$ edges that may be given as a certificate, at least one of the two conditions~(1) and~(2) checked by the verifier fails, thus the verifier rejects as desired. 
\end{proof}
    
\begin{lemma}\label{lem:another-MM-at-least-k}
    There exists a semi-streaming certification scheme for $\MM_{\geq}$. Specifically, in $n$-node graphs with maximum degree~$\Delta\geq 2$, the certificate size of the scheme is $O(n\log \Delta)$ bits, and the space complexity of the algorithm is $O(n)$ bits. 
\end{lemma}

\begin{proof}
Let $G=(V,E)$ be a graph with  maximum matching matching of size at least~$k$. The prover produces a coloring $c$ of the vertices of $G$ as a certificate, based on a lemma by Bousquet, Feuilloley, and Zeitoun~\cite[Lemma~28]{DBLP:Feuilloley}. 
The original lemma concerns perfect matchings, and we next state and prove an adaptation of it to arbitrary maximum matchings.

For every $n$-node graph $G=(V,E)$ with maximum degree $\Delta\geq2$, and for every matching $M$ in~$G$, there exists a coloring $c: V\rightarrow \{1,\dots,2 \Delta-1\}$ such that (1)~for every $\{u,v\}\in M$, $c(u) = c(v)$, and (2)~for every $\{u,v\}\in E\smallsetminus M$, $c(u) \neq c(v)$. 
To see why such a coloring exists, order the edges of $M$ arbitrarily as $e_1,\ldots, e_k$, and for every $i\in \{1,\ldots, k\}$, let $V(e_1,\ldots, e_i)$ be the set of vertices incident to the edges $e_1,\ldots,e_i$.
We first color the vertices $V(e_1,\ldots, e_i)$, for $i=1,\ldots, k$:
assume that the vertices $V(e_1,\ldots,e_{i-1})$ have already been colored appropriately (i.e., so that conditions~(1) and~(2) holds in $V(e_1,\ldots,e_{i-1})$), and consider $e_i=\{u,v\}$. 
The vertices $u$ and $v$ are assigned a (same) color that has not been adopted by any vertex in $(N(u)\cup N(v))\smallsetminus \{u,v\}$, where $N(\cdot)$ denotes the set of neighboring vertices of a vertex. 
Such a color exists merely because $|(N(u)\cup N(v))\smallsetminus \{u,v\}|\leq 2\Delta - 2$. 
Once the vertices of $V(e_1,\ldots, e_i)$ have been colored, the remaining vertices of $G$ can be properly colored greedily in an arbitrary order, as each vertex has at most $\Delta$ neighbors which are already colored, and $2\Delta-1\geq \Delta+1$ guarantees that one color is always available. 

So, according to the above, let $c: V\rightarrow \{1,\dots,2 \Delta-1\}$ be a coloring of $G=(V,E)$ such that (1)~for every $v\in V$, $|\{u\in N(v) \mid c(u) = c(v)\}|\leq 1$, and (2)~$|\{u,v\}\in E \mid c(u) = c(v)\}|\geq k$. This coloring is the certificate produced by the prover. The verifier aims at checking that conditions~(1) and~(2) hold. It uses one flag bit ${\rm flag}(i)$ per vertex $i\in [n]$, all initialized to~0. Upon reception of an edge $\{i,j\}$, if $c(i)=c(j)$ then the flag bits of $i$ and $j$ are flipped to~1. If the flag bit of at least one of the vertices is already flipped to~1, then the verifier rejects. At the end of the stream, if $\frac{1}{2} \sum_{i\in [n]}{\rm flag}(i) < k$, then the verifier rejects. Otherwise, it accepts. 

Again, completeness holds by construction. 
For soundness, let $G$ be a graph with maximum matching size smaller than~$k$. Towards a contradiction,  assume that there is a certificate, i.e., a coloring~$c$ such that the verifier accepts. As the verifier accepts, the monochromatic edges w.r.t.~$c$ form a matching whose cardinality is at least~$k$ which  yields a contradiction.

The coloring $c$ can be encoded on $O(n\log \Delta)$ bits. The flags consume exactly $n$ bits, and determining whether $\frac{1}{2} \sum_{i\in [n]}{\rm flag}(i) < k$ consumes $O(\log k)$ bits of space. Therefore, the space complexity of the verifier is  $O(n)$ bits.
\end{proof}

\begin{lemma}\label{lem:MM-at-most-k}
     There exists a semi-streaming certification scheme for $\MM_{\leq}$ with certificate size $O(n)$ bits, and verifier's space complexity $O(n\log n)$ bits.
\end{lemma}

\begin{proof}
We describe the prover and verifier of the certification scheme as follows. Let $G=(V,E)$ be a graph whose maximum matching size is at most~$k$. The scheme follows the same line of reasoning as in~\cite{Censor-HillelPP20} and~\cite{ChakrabartiG19}. Recall that the Tutte-Berge formula~\cite{DBLP:BondyM08} states that if a graph $G=(V,E)$ has a maximum matching of size~$k'$ then 
\[
k' = \frac{1}{2}\; \min\limits_{U\subseteq V}\Big(|U| - \mbox{odd}(V \smallsetminus U)+|V|\Big),
\]
where $\mbox{odd}(V \smallsetminus U)$ is the number of connected components with odd number of vertices in $G[V\smallsetminus U]$, i.e., the subgraph of $G$ induced by the vertices in $V \smallsetminus U$. Therefore, there exists a set $U\subseteq V$ such that 
\begin{equation}\label{eq:certificate_matching}
       k \geq \frac{1}{2}\; \Big(|U| - \mbox{odd}(V \smallsetminus U)+|V|\Big)
\end{equation}
The prover gives the set $U$ as a certificate, encoded as an $n$-bit vector.
%~$x\in\{0,1\}^n$ where, for every $i\in [n]$, $x[i] = 1$ if and only if $i\in U$. 
The verifier maintains a spanning forest $F$ of $G[V\smallsetminus U]$: 
as an edge $\{u,v\}$ streams in, the verifier keeps $\{u,v\}$ if and only if $u,v\notin U$ and it does not induce a cycle in~$F$.
As $F$ has $O(n)$ edges, the space required to store $F$ is $O(n\log n)$ bits (i.e., $O(\log n)$ bits per edge).  Once the stream ends, each tree in the  forest $F$ spans a connected component of $G[V\smallsetminus U]$. The verifier then counts the number of trees with odd number of vertices, accepts if Eq.~\eqref{eq:certificate_matching} holds, and rejects otherwise. 

If  the  maximum matching of $G$ is of size  at most~$k$, then there exists a set $U\subseteq V$ such that Eq.~\eqref{eq:certificate_matching} holds, and completeness follows. For soundness, assume that the size of a maximum matching in $G$ is larger than $k$, yet there exists a certificate $c$ making the verifier  accept $G$. 
This implies that Eq.~\eqref{eq:certificate_matching} holds for the set $U$ encoded in $c$, which, by the Tutte-Berge formula, implies that any maximum matching in~$G$ has size at most $k$, a contradiction.
\end{proof}

Lemmas~\ref{lem:MM-at-least-k}, \ref{lem:another-MM-at-least-k} and~\ref{lem:MM-at-most-k} together prove Theorem~\ref{theo:MM:upper-bound}. We now show that the streaming certification schemes provided in Theorem~\ref{theo:MM:upper-bound} for $\MM_{\geq}$ are essentially optimal, up to polylogarithmic factors. 

\begin{theorem}\label{thm:MM_lb}
    Any streaming certification scheme for $\MM_{\geq}$ using $c$-bit certificates and using $m$ bits of work-memory for $n$-node graphs satisfies $c+m=\Omega(n)$. 
\end{theorem}

\begin{proof}

The proof is via reduction to the communication problem ${\disj}^N_{N/2}$. In this problem, Alice is given set $x$ and Bob, the set $y$ as inputs respectively, which are subsets of $[N]$ of size $N/2$. The aim is to decide whether $x\cap y =\varnothing$. It was proven in~\cite{RaoY2020} that the non-deterministic communication complexity of ${\disj}_N$ is $\Omega(N)$. The same line of argument extends to non-deterministic complexity of ${\disj}^N_{N/2}$ as well. Let $x = \{x_1,\ldots, x_{N/2}\}$ and $y = \{y_1,\ldots, y_{N/2}\}$. Alice and Bob create a bipartite graph $G = (L\cup R, E )$ as follows. The vertex sets $L$ and $R$ are defined as follows. 
\begin{align*}
	L \coloneqq \{u_1,\ldots, u_{N}\}, R \coloneqq \{v_1,\ldots, v_{N}\}
\end{align*}
Alice and Bob add edges $E_A$ and $E_B$ respectively such that the edge set $E$ satisfies $E_A\cup E_B$. The edge sets $E_A$ and $E_B$ are defined respectively as follows: for each $i\in [N/2]$, $(u_{x_i},v_i)\in E_A$ and $(u_{y_i},v_{N/2+i})\in E_B$. By the construction of $G$, it has $N$ edges and if $x\cap y = \varnothing$, then $G$ is a perfect matching. On the other hand, if $x\cap y \neq \varnothing$, then there exists a vertex $u\in L$ and two distinct vertices $v_i$ and $v_j$ in $R$ such that $i\leq N/2$ and $j>N/2$ and $u$ is adjacent to both $v_i$ and $v_j$.  Since $G$ has only $N$ edges, there exists an isolated vertex in $L$ and therefore, $G$ does not have a perfect matching. Hence, $G$ has a perfect matching if and only if $x\cap y = \varnothing$.  

If there exists a streaming certification scheme for $\MM_{\geq}$ on a $N$-vertex graph having $c$-bit certificate and $m$-bit verifier space satisfying that $c+m = o(N)$, then given an instance $(x,y)$ of $\disj_N^{N/2}$, Alice and Bob create the graph $G$ as discussed above, Alice simulates the certification scheme for $k = N/2$ on the arbitrarily ordered edges of $E_A$, sends the state space of the verifier to Bob and Bob continiues the computation on the remaining edges $E_B$ and accepts if and only if the algorithm decides that $G$ has a perfect matching, thus describing a $o(N)$ non-deterministic communication protocol for $\disj_N^{N/2}$ which leads to a contradiction.     
\end{proof}

%----------------------------------------------------------------
\subsection{Degeneracy}
%----------------------------------------------------------------

Recall that, for every integer $k\geq 0$, a $k$-degenerate graph is a graph in which every subgraph has a vertex of degree at most~$k$. The \emph{degeneracy} of a graph is the smallest value of~$k$ for which it is $k$-degenerate. Equivalently, a graph is $k$-degenerate if and only if the edges of the graph can be oriented to form a directed acyclic graph with out-degree at most~$k$~\cite{DBLP:BondyM08}. 

\begin{theorem}\label{theo:certif-degeneracy-upper-bounds}
    There exists a semi-streaming certification scheme for $\Degeneracy$. Specifically:
    \begin{itemize} 
        \item There exists a semi-streaming certification scheme for $\Degeneracy_{\leq}$ with certificate size $O(n\log n)$ bits, and verifier's space complexity $O(n\log k)$ bits for all $n$-node graphs and thresholds $k\geq 1$. 
        \item There exists a semi-streaming certification scheme for $\Degeneracy_{\geq}$ with certificate size $\min\{k\log n, n\}$ bits, and verifier's space complexity $O(n\log k)$ bits  for all $n$-node graphs and thresholds $k\geq 1$.  
    \end{itemize}  
\end{theorem}

The theorem directly follows from lemmas~\ref{lem:degeneracy-at-most-k-upper-bound} and~\ref{lem:degeneracy-at-least-k-upper-bound} hereafter. 

\begin{lemma}\label{lem:degeneracy-at-most-k-upper-bound}
There exists a semi-streaming certification scheme for $\Degeneracy_{\leq}$ with certificate size $O(n\log n)$ bits, and verifier's space complexity $O(n\log k)$ bits.
\end{lemma}

\begin{proof}
  An $n$-node graph $G$ with degeneracy at most~$k$ can be acyclically oriented such that the out-degree of each vertex is at most $k$. The prover provides the topological ordering $\pi$ of $V(G)$ of one such acyclic orientation of $G$ as a certificate, which can be encoded on $O(n\log n)$ bits. Given a certificate $\pi$, the verifier uses $n$ counters, one per vertex of~$G$, all initialized to~0. For each edge $\{u,v\}$ that is streamed in, the verifier increments the counter of $\min\{\pi(u),\pi(v)\}$ by~$1$. At the end of the stream, the verifier accepts if and only if no counter exceeds~$k$. 

    Completeness is satisfied by construction. To argue for soundness, let us assume that $G$ has degeneracy larger than~$k$. Therefore, for every acyclic orientation of $G$, there exists a vertex $u$ whose out-degree is greater than $k$. Therefore, for any ordering $\pi$ of its vertices, there exists a vertex $u$ such that $|\{v\in N(u)\ :\ \pi(v) < \pi(u)\}|> k$. Given $\pi$ as a certificate, the counter of~$u$ will exceed~$k$, leading the verifier to reject as desired.  
\end{proof}

\begin{lemma} \label{lem:degeneracy-at-least-k-upper-bound}
There exists a semi-streaming certification scheme for $\Degeneracy_{\geq}$ with certificate size $O(\min\{k\log n, n\})$ bits, and verifier's space complexity $O(n\log k)$ bits. 
\end{lemma}

\begin{proof}
Let $G$ be a graph with degeneracy at least~$k$. For the sake of convenience, we assume that $V = [n]$. There exists a an induced subgraph $G[V']$ of $G$ such that, for every $v\in V'$, $\deg_{G[V']}(v)\geq k$. The prover provides $V'$ to the verifier as certificate, which is provided as a list of vertices (requiring $k\log n$ bits) or be encoded as an $n$-bit vector (depending on whether $k\log n < n$). Given certificate $V'$, the verifier checks whether $\deg_{G[V']}(v)\geq k$ for every $v\in V'$. For this, it is sufficient to store counter for each vertex in $V'$, initialized to~$0$. For each edge $\{u,v\}$ that streams in, if $\{u,v\}\subseteq V'$ then the verifier increments the counters of both vertices $u$ and $v$ by~$1$. For saving space, a counter with value~$k$ is not incremented. At the end of the stream, if the counters of all the vertices in $V'$ have reached the value~$k$, then the verifier accepts, else it rejects. Note that storing the counters solely requires $O(n\log k)$ bits.

Again, completeness is satisfied by construction. To argue for soundness, let $G$ be a graph with degeneracy less than~$k$. For every induced subgraph $G[V']$ of $G$, there exists a vertex $v\in V'$ such that $\deg_{G[V']}(v)< k$. Therefore, for any set of vertices $V'$ provided by the prover, the verifier will reject, as the counter of~$v$ won't reach~$k$.
\end{proof}

 We now show that the certification schemes provided in Theorem~\ref{theo:certif-degeneracy-upper-bounds} are essentially optimal, up to polylogarithmic factors by reduction from disjointness. This holds even for a fixed threshold~$k=1$.

\begin{theorem}
    Any streaming certification scheme for $\Degeneracy_{\leq 1}$ consuming $m$-bit memory with $c$-bit certificates satisfies $c+m=\Omega(n)$ bits in $n$-node graphs.   
\end{theorem}

\begin{proof}
The proof is via reduction from $\disj_N$. Given an instance $(x,y)$ of $\disj_N$ where $x$ and $y$ are subsets of $[N]$, let $x = \{x_1,\ldots, x_t\}$ and $y = \{y_1,\ldots,y_l\}$. Alice and Bob create a graph $G = ([N]\cup \{a,b\}, E_A\cup E_B)$ where Alice adds the edges $E_A$ and Bob adds the edges $E_B$ as follows. Alice adds the edges such that $E_A$ forms a star graph on the vertex set $\{a,b\}\cup x$ with $a$ as the root and $b$ and the set $x$ forming the leaves. Bob adds the edges such that $E_B$ forms an arbitrary path on the the vertex set $\{b\}\cup y$ with one of the endpoints of the path being $b$. By the construction of $G$, if $x\cap y = \varnothing$, $G$ is a tree and therefore, has degeneracy $1$. Else, if $x\cap y \neq \varnothing$, let $w$ be an element that is both in $x$ and $y$. Then, there exists two distinct paths between $w$ and $b$ in $G$, namely $\{w,a,b\}$ in $E_A$ and the path connecting $w$ and $b$ in $E_B$. Therefore, $G$ is at least $2$-degenerate. Hence, $G$ is $1$-degenerate if and only if $x\cap y = \varnothing$. 

If there exists a streaming certification scheme for $\Degeneracy_{\leq 1}$ on a $N$-vertex graph having $c$-bit certificate and $m$-bit verifier space satisfying that $c+m = o(N)$, then given an instance $(x,y)$ of $\disj_N$, Alice and Bob create the graph $G$ as discussed above, Alice simulates the certification scheme for $k = 1$ on the arbitrarily ordered edges of $E_A$, sends the state space of the verifier to Bob and Bob continiues the computation on the remaining edges $E_B$ (which are arbitrarily ordered as well) and accepts if and only if the algorithm decides that $G$ has degeneracy at most~$1$, thus describing a $o(N)$ non-deterministic communication protocol for $\disj_N$ which leads to a contradiction.     
\end{proof}

%----------------------------------------------------------------
\subsection{Diameter}
%----------------------------------------------------------------

In contrast to Theorem~\ref{theo:lower-bounds-decision}, we show that $\Diam_{\geq}$ can be certified in semi-streaming.

\begin{theorem}\label{theo:diam_atleast_k}
    There exists a semi-streaming certification scheme for $\Diam_{\geq}$. Specifically,  the size of the certificates provided by the prover is $O(n\log k)$ bits, and the space complexity of the verifier is $O(\log n)$ bits for all $n$-node graphs, and threshold $k\geq 1$.
\end{theorem}

\begin{proof}
 Let $k\geq 1$, and let $G=(V,E)$ be a graph with diameter at least~$k$. 
 There exist two vertices $u,v\in V$ at distance $k$ in~$G$. 
 One solution is for the  prover to provide the verifier with a BFS tree $T$ rooted at $u$ as certificate. 
 Such a certificate however requires $O(n\log n)$ bits to be encoded. 
 Instead, the prover provides the verifier with the distance of every node $w$ to $u$, with the an exception that if $\mbox{dist}(u,w)\geq k+1$, then the prover gives distance~$k+1$.
 These distances can be encoded on $O(n\log k)$ bits.  
 
 The verifier first checks that a node with certificate $0$ and a node with certificate at least $k$ exist, and rejects otherwise.
 When an edge $e=\{x,y\}$ is streamed in, the verifier checks that $e$ does not create a shortcut, i.e., it checks that $|\mbox{dist}(u,x)-\mbox{dist}(u,y)|\leq 1$. 
 At the end of the stream, if no shortcuts have been streamed, then the verifier accepts, else it rejects. 
 The space complexity  of the verifier is $O(\log n)$ bits.
 Note that, in fact, certain certificates where $\mbox{dist}(u,x)$ is lower than the real distance between $u$ and $x$ might also be accepted by the verifier.
 Nevertheless, this is enough in order to verify the existence of a node $v$ with distance at least $k$ from $u$, as detailed below.

 Completeness is satisfied by construction. 
 For soundness, let $G$ be a graph of diameter less than~$k$, and let $\mbox{dist}$ be the certificate provided by the prover. Assume $G$ is accepted, then there is a node $u$ with certificate $\mbox{dist}(u,u)=0$, and a node $v$
 with certificate $\mbox{dist}(u,v)\geq k$.
 Let~$P$ be a shortest path from $u$ to $v$, which has length smaller than $k$ by assumption.
 Then, any assignment of certificates to the nodes of $P$ must have two consecutive nodes $x,y$ whose certificates 
  $\mbox{dist}(u,x), \mbox{dist}(u,y)$
 differ by more than 1, and the verifier will reject upon the arrival of the edge~$\{x,y\}$.
\end{proof}

 We now show that the certification schemes provided in Theorem~\ref{theo:diam_atleast_k} are essentially optimal, up to polylogarithmic factors by reduction from disjointness. This holds even for a fixed threshold $k=8$.

\begin{theorem}
    Any streaming certification scheme for $\Diam_{\geq 8}$ consuming $m$-bit memory with $c$-bit certificates satisfies $c+m=\Omega(n)$ bits in $n$-node graphs.   
\end{theorem}

\begin{proof}
The proof is via reduction from $\disj_N$. Given an instance $(x,y)$ of $\disj_N$ where $x$ and $y$ are subsets of $[N]$, let $x = \{x_1,\ldots, x_t\}$ and $y = \{y_1,\ldots,y_l\}$. Alice and Bob create a graph $G = (P_1\cup P_2\cup P_3 \cup \{u,v,a,b, t_1,t_2, t_3,t_4\}, E_A\cup E_B)$ where Alice and Bob add edges $E_A$ and $E_B$ respectively. The graph $G$ is defined as follows. The vertex sets $P_1$, $P_2$ and $P_3$ are defined as follows. For $i \in [3]$, $P_i = \{v^i_1,\ldots,v^i_N\}$. Alice adds the edge set $E_A$ as follows. She adds the edges $\{u,a\}$, $\{b,v\}$, $\{t_1,t_2\}$ and $\{t_3,t_4\}$ to the graph $G$. She also adds $\{a,v^1_i\}$, $\{v^3_i,b\}$, $\{v^1_i, t_1\}$, $\{t_2, v^2_i\}$, $\{v^2_i,t_3\}$ and $\{t_4, v^3_i\}$ for each $i\in [N]$. For each $i\in x$, she also adds the edge $\{v^1_i,v^2_i\}$. Bob adds the edge set $E_B$ by adding the edge $\{v^2_j,v^3_j\}$ for each $j\in y$. By the construction of $G$, the diameter of $G$ is the length of the shortest path between $u$ and $v$. If $x\cap y\neq \varnothing$, there exists an element $i\in x\cap y$ and therefore, a path of length $6$ between $u$ and $v$ through $a,v^1_i, v^2_i, v^3_i$ and $b$. Else if $x\cap y = \varnothing$, for any element $i\in x$, there is a path between $u$ and $v^2_i$ via $a$ and $v^1_i$ but as $i\notin y$, the shortest path between vertices $v^2_i$ and $v^3_i$  is via vertices $t_3$ and $t_4$, and therefore, the shortest path between $u$ and $v$ is at least $8$. Therefore $G$ has diameter at least $8$ if and only if $x\cap y = \varnothing$.   

If there exists a streaming certification scheme for $\Diam_{\geq 8}$ on a $N$-vertex graph having $c$-bit certificate and $m$-bit verifier space satisfying that $c+m = o(N)$, then given an instance $(x,y)$ of $\disj_N$, Alice and Bob create the graph $G$ as discussed above, Alice simulates the certification scheme for $k = 8$ on the arbitrarily ordered edges of $E_A$, sends the state space of the verifier to Bob and Bob continiues the computation on the remaining edges $E_B$ (which are arbitrarily ordered as well) and accepts if and only if the algorithm decides that $G$ has diameter at least $8$, thus describing a $o(N)$ non-deterministic communication protocol for $\disj_N$ which leads to a contradiction.            
\end{proof}

%----------------------------------------------------------------
\subsection{Colorability}\label{sec:Colorability}
%----------------------------------------------------------------

As exemplified in Section~\ref{sec:intro}, there is a trivial streaming certification scheme for $\Col_{\leq}$ where, for every input graph $G$, and every input threshold $k\geq 1$ with $\chi(G)\leq k$, the prover merely provides the verifier with the color of each vertex of $G$ is a proper $k$-coloring of $G$. This certificates can be encoded on $O(n\log k)$ bits. We show that this is optimal. 

\begin{theorem}
    Any streaming certification scheme for $\Col_{\leq}$ consuming $m$-bit memory with $c$-bit certificates satisfies that $c+m = \Omega(n\log n)$ bits.
\end{theorem}

\begin{proof}
The proof is via reduction from the communication problem $\perm_r$ where $r\geq 3$. In this problem, Alice and Bob are given two permutations $\sigma$ and $\tau$ over $[r]$ respectively. The aim is to decide whether $\sigma = \tau$. First, we prove the following lower bound for non-deterministic communication complexity for $\perm_r$.  
Let $\CCN(f)$ denote the non-deterministic communication complexity of a given communication problem~$f$.  

\begin{fact}\label{comm_lower_bound:perm}
	$\CCN(\perm_r) = \Omega(r\log r)$.
\end{fact}
\begin{proof}
	Let the non-deterministic string that Alice and Bob receive from the oracle upon an arbitrary input $(\sigma, \tau)$ where Alice and Bob receive $\sigma$ and $\tau$ respectively, is of length at most $k$. If $k = o(r\log r)$, then there exist two distinct input instances $(\sigma',\sigma')$ and $(\tau',\tau')$ such that Alice and Bob receive the same non-deterministic string $s$ from the oracle as $2^k<r!$. Let the third instance be $(\sigma',\tau')$. Alice and Bob accept the third instance with the  non-deterministic string $s$ as they cannot differentiate between $(\sigma',\tau')$ and $(\sigma',\sigma')$ or $(\tau',\tau')$ which leads to a contradiction. 
\end{proof}
Now we describe the reduction to $\kCol_{\leq}$ using a construction borrowed from~\cite{DBLP:Feuilloley}. Given an input instance $(\sigma,\tau)$ of $\perm_k$, Alice and Bob create the graph $G$ such that $\chi(G)\leq k$ if and only if $\sigma = \tau$. Before describing the graph $G$, we describe the components of $G$ as follows. Let $P$ be the graph with vertex set $V(P) =[k]\times [2]$, and edge set
    \[
    E(P) = 
    \big \{\{(i,p),(j,p)\}\mid (i,j)\in [k]\times [k] \land i\neq j \land p\in [2]\big  \}
    \cup
    \big \{\{(i,1),(j,2)\} \mid (i,j)\in [k]\times [k] \land i\neq j\big \}.
    \]    
    Let $P$ and $P'$ be two copies of this graph, and let $E_{\sigma}$ and $E'_{\tau}$ be two sets of edges defined as follows. 
    \begin{align*}
        E_{\sigma} &=\{\{(i,1),(j,1)\} \mid (i,1)\in V(P) \land (j,1)\in V(P') \land j\neq \sigma(i)\}, \text{ and} \\
        E'_{\tau} &=   \{\{(i,2),(j,2)\} \mid (i,1)\in V(P) \land (j,1)\in V(P') \land j\neq \tau(i)\}.
    \end{align*}
The graph $G$ is defined as $G = P\cup P'\cup E_{\sigma}\cup E'_{\tau}$. Alice adds the edges of the graphs $P$ and $P'$ as well as $E_{\sigma}$ and Bob adds the edges $E'_{\tau}$.
\begin{fact}[Claim~18 in~\cite{DBLP:Feuilloley}] \label{claim:colorability}
   $G$ is $k$-colorable if and only if $\sigma = \tau$.
\end{fact} 

If there exists a certification scheme for $\kCol_{\leq}$ for $n$-vertex graphs taking $c$-bits of certificate and $m$-bit verifier space such that $c+m = o(n\log n)$, then given an instance $(\sigma,\tau)$ of $\perm_{r}$, Alice and Bob create graph $G$ as discussed above and as $G$ has $2r$ vertices, Alice simulates the certification scheme for $k = r$ on her part of the graph and sends the state space of the verifier to Bob who continues the simulation on his edges and accepts if and only if the algorithm decides that $G$ is $r$-colorable. But this describes a $o(r\log r)$-protocol for $\perm_r$ which leads a contradiction to Fact~\ref{comm_lower_bound:perm}.
\end{proof}

%--------------------------------------------------------------
\subsection{Clique, independent set, and vertex cover}
\label{sec:mincover-easy-direction}
%--------------------------------------------------------------

We conclude this section about ``easily certifiable problems'' by considering maximum independent set ($\MaxIS$), maximum clique ($\MaxClique$), and minimum vertex cover ($\MinCover$). The proof of the following result is using similar arguments as those previously presented in this section. 

\begin{theorem}\label{thm:easy-direction-IS-Clique-etc}
$\MaxIS_{\geq}$, $\MaxClique_{\geq}$, and $\MinCover_{\leq}$ can all be certified with certificates on $O(\min\{k\log n, n\})$ bits, and verifier's space complexity $O(\log n)$ bits  for all $n$-node graphs and threshold $k \geq 1$.
\end{theorem}

\begin{proof}
    The certificate for $\MaxIS_{\geq}$ is an IS $S$ of size $k$ where $k$ is the input threshold, and the verifier merely checks that no stream edges are in $S\times S$.
    The certificate for $\MaxClique_{\geq}$ is similarly a clique $S$ of size $k$, which can be verified by counting the edges in $S\times S$ edges in the stream (there can be no more than $n^2$ edges in the clique, and thus a counter on $O(\log n)$ bits suffices.
    For $\MinCover_{\leq}$, a VC $S$ of size at most $k$ is a certificate, which is accepted unless an edge not touching $S$ arrives in the stream.
\end{proof}

%%%%%%%%%%%%%%%%%%%%%%%%%%%%%%%%%%%%%%%%%%%%%%%%%%%%%%%%%%%%%%%%%
\section{Hard Problems for Streaming Certification}\label{sec:Hard problems}
%%%%%%%%%%%%%%%%%%%%%%%%%%%%%%%%%%%%%%%%%%%%%%%%%%%%%%%%%%%%%%%%%

Not all problems can be certified in the semi-streaming model. This section exhibits a series of classical problems essentially requiring $\Omega(n^2)$ space complexity for being certified. Our lower bounds follows the proof strategy developed in~\cite{AbboudCKP21}, which define so-called \emph{families of streaming lower bound graphs}. Given a function $f:\{0,1\}^N\times\{0,1\}^N\to\{0,1\}$ and a graph predicate $\prob$, a family of graphs 
\[
\mathcal{G}=\{G_{x,y}=(V,E\cup A_{x}\cup B_{y})\mid (x,y)\in \{0,1\}^N\times\{0,1\}^N\}
\]
is a   family of streaming lower bound graphs for $f$ and $\prob$ if $A_x$ depends only on~$x$, $B_y$ depends only on~$y$, and $$G\models \prob \iff f(x,y)=1.$$ The following lemma is an extension of Theorem~15 in~\cite{AbboudCKP21} to non-deterministic communication complexity. 

\begin{lemma}\label{lem:family-of-streaming-lower-bound}
    If there exists a family of streaming lower bound graphs for $f$ and $\prob$ then any non-deterministic streaming certification scheme for $\prob$ with $c$-bit certificates and $m$ bits of memory satisfies $c+m=\Omega(\CCN(f))$. 
\end{lemma}

\begin{proof}
    Given an instance $(x,y)$ of $f$, Alice and Bob can construct $A_{x}$ and $B_{y}$, respectively. Alice can then simulate the non-deterministic certification scheme by streaming the edges of $E\cup A_{x}$ in arbitrary order, and then send to Bob the state of the verifier after all edges in $E\cup A_{x}$ has been streamed. Then Bob can carrying on the verification by streaming all edges in~$B_y$. Bob then outputs~0 or~1, according to whether the verifier rejects or accepts, respectively. 
\end{proof}

%----------------------------------------------------------------
\subsection{Diameter}
%----------------------------------------------------------------

While the certification of $\Diam_{\geq}$ was proved to be possible in semi-streaming, this is not the case of $\Diam_{\leq 2}$, i.e., even for a threshold  $k=2$.  

\begin{theorem}\label{theo:diam_lb}
    Any streaming certification scheme for $\Diam_{\leq 2}$ consuming $m$-bit memory with $c$-bit certificates satisfies $c+m = \Omega(n^2)$ bits in $n$-node graphs.
\end{theorem}

\begin{proof}
    A family of streaming lower bound graphs for $\disj_N$ and $\Diam_{\leq 2}$ was given in~\cite{HolzerW12}. Let $N=p(p-1)/2$. Each graph has $2(p+1)$ nodes, forming the set $V=\{a_0,\dots,a_p,b_0,\dots,b_p\}$. Let 
    $$
    E= \{\{a_i,b_i\}, i=0,\dots,p\} \cup \{\{a_0,a_i\}, i=1,\dots,p\}\cup \{\{b_0,b_i\}, i=1,\dots,p\}.$$ 
    Given $(x,y) \in\{0,1\}^N\times \{0,1\}^N$, for every $(i,j)\in [p]\times [p]$ with $i<j$, we set  $\{i,j\}\in A_x$ if $x[i,j]=0$, and $\{i,j\}\in B_y$ if $y[i,j]=0$. It was shown in~\cite{HolzerW12} that $G=(V,E\cup A_x\cup B_y)$ has diameter~2 if and only if $x\cap y=\varnothing$. Hence, there exists a family of streaming lower bound graphs for $\disj$ and $\Diam_{\leq 2}$. The result follows by direct application of Lemma~\ref{lem:family-of-streaming-lower-bound}. 
\end{proof}

%--------------------------------------------------------------
\subsection{Minimum vertex cover and maximum independent set}
\label{sec:mincover-hard-direction}
%--------------------------------------------------------------

\begin{theorem}\label{thm:hard-mini-cover}
Any streaming certification scheme  for $\MinCover_{\geq}$
consuming $m$-bit memory with $c$-bit certificates satisfies $c+m=\Omega(n^2)$ in $n$-node graphs.
\end{theorem}

\begin{proof}
The proof is by reduction from $\disj$  from~\cite{AbboudCKP21}. Small details have to be fixed for adapting this reduction to our setting, so we recall their construction hereafter. Let us first recall the Bit Gadget introduced in~\cite{AbboudCKP21}. Let $N$ be a power of $2$, and let 
\[
A = \{a_i \mid  i\in \{0,\ldots,N-1\}\}, \text{ and } 
B = \{b_i \mid  i\in \{0,\ldots,N-1\}\}.
\]
For each set $S\in \{A,B\}$, let 
\[
F_S = \{f_S^i\ |\ i\in \{0,\ldots,\log N -1\}\}, \text{ and } 
T_S =\{t_S^i\ |\ i\in \{0,\ldots,\log N -1\}\},
\]
where $T$ and $F$ stands for true and false, respectively. We consider the graph  with vertex set 
$
A\cup B\cup F_A \cup F_B\cup T_A \cup T_B.
$
Let $S\in \{A,B\}$, and let $s_i\in S$. 
Let  $bin(s_i)=\iota_{N-1}\dots \iota_0$ be the binary representation of $i$ on $\log N$ bits. 
Then, for every $j\in\{0,\dots,N-1\}$, connect $s_i$ to $f_S^j$ if $\iota_j=0$, and to $t_S^j$ if $\iota_j=1$. Finally, for every $i\in \{0,\ldots, \log N -1\}$ connect $f_A^i$ to $t_B^i$, and $t_A^i$ to $f_B^i$. This completes the description of the Bit Gadget.

The construction carries on by taking two copies of the Big Gadget, that is, there are four sets of $N$ vertices, denoted by
\[
A = \{a_i \mid i\in \{0,\ldots, N-1\}\}, \text{ and } B = \{b_i \mid i\in \{0,\ldots, N-1\}\},
\]
as above, and by 
\[
A' = \{a'_i \mid i\in \{0,\ldots, N-1\}\}, \text{ and } B' = \{b'_i \mid i\in \{0,\ldots, N-1\}\}.
\]
Each such set $S\in\{A,B,A',B'\}$ is connected to the $2\log N$ nodes in $F_S = \{f_i^S\ |\ i\in \{0,\ldots, \log N-1\}\}$ and $T_S = \{t_i^S\ |\ i\in \{0,\ldots,\log N-1\}\}$ as in the Big Gadget, where $A,B,F_A,T_A,F_B,T_B$ is one gadget, and $A',B',F_{A'},T_{A'},F_{B'},T_{B'}$ is the other gadget. 

In addition, for each $S\in\{A,B,A',B'\}$, connect all the nodes in $S$ to form a clique.  Also, connect the nodes of the bit-gadgets to form a collection of $4$-cycles, as follows. For every  $i\in \{0,\ldots, \log N -1\}$, there is a $4$-cycle $(f_{A}^i,t_{A}^i,f_{B}^i,t_{B}^i)$, and a $4$-cycle $(f_{A'}^i,t_{A'}^i,f_{B'}^i,t_{B'}^i)$. Let $G$ be the resulting graph. It was shown in~\cite{AbboudCKP21} that every vertex cover of $G$ contains at least $N-1$ nodes from each of the clique  $S\in\{A,B,A',B'\}$, and at least $4\log N$ in $F_A \cup T_A  \cup F_B  \cup T_B  \cup F_{A'}  \cup T_{A'}  \cup F_{B'}  \cup T_{B'}$.  Moreover, a vertex cover of $G$ of size $4(N-1) + 4\log N$ exists the if the condition in the following fact holds. 

\begin{fact}[Claim 6 in \cite{AbboudCKP21}]\label{fact:6th-claim-of-Abboud-et-al}
    Let $U\subseteq V(G)$ be a vertex cover of $G$ of size $4(N-1) + 4\log N$. There exists $(i,j)\in \{0,\ldots, k-1\}^2$ such that $U\cap \{a_i,b_i,a'_j,b'_j\}=\varnothing$.
\end{fact}

Now, let $x\in\{0,1\}^{N^2}$ and $y\in\{0,1\}^{N^2}$ be an instance of $\disj$, where Alice receives $x$ as input, 
Bob receives $y$ as input, and the elements of $x$ and $y$ are indexed by the pairs $(i,j)\in \{0,\ldots,N-1\}^2$.
Let $G_{x,y}$ be the graph $G$ augmented with some additional edges, as follows. For every $(i,j)\in \{0,\ldots,N-1\}^2$, if $x[i,j] = 0$ then add the edge $\{a_i,a'_j\}$, and if $y[i,j] = 0$ then add the edge $\{b_i,b'_j\}$.

\begin{fact}[Lemma 4 in~\cite{AbboudCKP21}]\label{fact:disjointness criterion}
    The graph $G_{x,y}$ has a vertex cover of size $4(N - 1) +4\log N$ if and only if $x\cap y\neq \varnothing$. 
\end{fact}

Let us revisit this latter fact, for figuring out what is the size of a minimum vertex cover in the case where $x$ and $y$ are disjoint. 

\begin{fact}\label{fact:nex-fact-on-bit-gadget}
    The graph $G_{x,y}$ has a minimum vertex cover of size at least $4(N-1) + 4\log N + 1$ if and only if $x\cap y=\varnothing$.
\end{fact}

To establish this fact, let us assume that $x\cap y=\varnothing$. Then, let us assume, for the purpose of contradiction, that every minimum vertex cover $U$ of $G$ has size at most $4(N-1) + 4\log N$. By Fact~\ref{fact:6th-claim-of-Abboud-et-al}, there exists four nodes $a_i,b_i,a'_j,b'_j$ not in $U$. Therefore, $\{a_i,a'_j\}\notin E(G_{x,y})$ and 
$\{b_i,b'_j\}\notin E(G_{x,y})$, which implies that $x[i,j] = y[i,j] = 1$, contradicting $x\cap y=\varnothing$. Conversely, let us assume that the size of a minimum vertex cover of $G$ is at least $4(N - 1) + 4\log N + 1$. If $x\cap y\neq \varnothing$, then, thanks to Fact~\ref{fact:disjointness criterion}, $G_{x,y}$ has a vertex cover of cardinality $4(N-1) + 4\log N$, a contradiction.  This completes the proof of Fact~\ref{fact:nex-fact-on-bit-gadget}.

\medskip

We now have all ingredients to apply Lemma~\ref{lem:family-of-streaming-lower-bound}. Let $V_A=A\cup A'\cup F_{A}\cup T_{A}\cup F_{A'}\cup T_{A'}$ and $V_B=B\cup B'\cup F_{B}\cup T_{B}\cup F_{B'}\cup T_{B'}$. 
Alice computes $A_x$, the edges of $G_{x,y}$ induced by the nodes in $V_A$, as well as $E$, the edges of $G_{x,y}$ between $V_A$ and $V_B$. 
Similarly, Bob computes $B_y$, the edges of $G_{x,y}$ induced by the nodes in~$V_B$. 
If there exists a certification scheme for~$\MinCover_{\geq}$ consuming $m$-bit memory with $c$-bit certificates satisfying $c+m=o(N^2)$ for the family of graphs $\mathcal{F }_N = \{G_{x,y}\mid  x,y\in \{0,1\}^{N^2}\}$, then there exists a non-deterministic protocol for \disj\ with complexity $o(N^2)$ (as $N = \Theta(n)$), contradiction.
\end{proof}

In an $n$-node graph $G=(V,E)$, if $I$ is an independent set then $V\smallsetminus I$ is a vertex cover. 
 Hence, the sum of the $\MaxIS$ size and the $\MinCover$ size is $n$,
 and an certification scheme for 
 $\MaxIS_{\leq k}$ can also be used to certify
 $\MinCover_{\geq n-k}$.
 The following is thus a direct consequence of Theorem~\ref{thm:hard-mini-cover}. 
 
\begin{corollary}
    Any streaming certification scheme for $\MaxIS_{\leq}$ consuming $m$-bit memory with $c$-bit certificates satisfies $c+m=\Omega(n^2)$ in $n$-node graphs.
\end{corollary}

%--------------------------------------------------------------
\subsection{Maximum Clique}
%--------------------------------------------------------------
\begin{theorem}
    Any streaming certification scheme for $\MaxClique_{\leq 3}$ consuming $m$-bit memory with $c$-bit certificates satisfies $c+m = \Omega(n^2)$ in $n$-node graphs.
\end{theorem}

\begin{proof}
    The proof is by reduction from \disj. Before giving the reduction, we define the  \textit{lower bound graph} which was introduced in~\cite[Definition~1]{DBLP:CzumajK20} for detection of cliques in \cgst. 

Let $G = (A\cup B,E)$ be a bipartite graph with $|A| = |B| = n$ and let $k,m$ be integers. 
Then $G$ is called a \textit{$(k,m)$-lower bound graph} if $|E| \leq m$ and there exist bipartite graphs $H_A = (A,E_A)$ and $H_B = (B,E_B)$ with $E_A = \{e_1,\ldots,e_k\}$,
$E_B = \{f_1,\ldots,f_k\}$, 
and $|E_A| = |E_B| = k$ on the vertex sets $A$ and $B$ respectively, so that:
\begin{enumerate}
    \item The graph $G\cup \{e_i,f_i\}$ contains a $K_4$, for every $1\leq i\leq k$, and
    \item The graph $G\cup \{e_i,f_j\}$ does not contain a $K_4$, for every $1\leq i,j\leq k$ with $i\neq j$.
\end{enumerate}

\begin{fact}[Section~3 and Theorem~5 of~\cite{DBLP:CzumajK20}]
    For every $n$, there exists a $(\Omega(n^2),O(n^{3/2}))$-lower bound graph.
\end{fact}
Given $x,y\in \{0,1\}^{n^2}$, we create $G_{x,y}$ by taking an $(\Omega(n^2),O(n^{3/2}))$-lower bound graph $G$ and augmenting it with edge sets $A_x$ and $B_y$ which are defined as follows.
\begin{align*}
    &A_x = \{e_i\in E_A\ :\ x[i] =1\} \\
    &B_y = \{f_j\in E_B\ :\ y[j] =1\} .
\end{align*} 

The following fact is implied by the definition of the lower bound graph (see also the proof of Theorem~2 in~\cite{DBLP:CzumajK20}).
\begin{fact}
    $\MaxClique_{\leq 3}(G_{x,y}) = 1$ if and only if $\disj(x,y) = 1$
\end{fact}

If the sets are not disjoint, they intersect in some index $i$, and the graph contains a $K_4$ on the vertices of $e_i\cup f_i$.
On the other hand, since $H_A$ and $H_B$ are bipartite, any possible $K_4$ must contain only one edge of $E_A$ and one of $E_B$.
The structure of $G$ guarantees that the only way to get such a $K_4$ is with edges $e_i$ and $f_i$ of the same index $i$, i.e., in which case the inputs intersect at index $i$.

The theorem now follows immediately from Lemma~\ref{lem:family-of-streaming-lower-bound}.
\end{proof}
%--------------------------------------------------------------
\subsection{Colorability}
%--------------------------------------------------------------

\begin{theorem}\label{theo:coloring-is-hard}
    Any streaming certification scheme for $\Col_{\geq 3}$ consuming $m$-bit memory with $c$-bit certificates satisfies $c+m = \Omega(n^2)$ bits in $n$-node graphs.
\end{theorem}

\begin{proof}
Again, we use a construction from Section~5.2 (and Theorem~16) from~\cite{AbboudCKP21} where it is proved that there is a  family of streaming lower bound graphs for $\disj$ and $\Col_{\geq 3}$. The theorem follows from Lemma~\ref{lem:family-of-streaming-lower-bound}. 
\end{proof}

%%%%%%%%%%%%%%%%%%%%%%%%%%%%%%%%%%%%%%%%%%%%%%%%%%%%%%%%%%%%%%%%%
\section{Conclusions}
%%%%%%%%%%%%%%%%%%%%%%%%%%%%%%%%%%%%%%%%%%%%%%%%%%%%%%%%%%%%%%%%%

In this paper we have introduced a new tool for streaming computation, that we refer to as {\em streaming certification}, and we have demonstrated its applicability, concentrating on the decision versions of graph problems, and on the semi-streaming regime. 
Streaming certification was inspired by the \emph{annotated streaming model}~\cite{ChakrabartiCMT14}, and by work in the area of \emph{distributed certification}~\cite{FKP13,GS16,KKP10}. In the latter, the access to the input is limited not by memory restrictions (as in the case in streaming) but merely because it is distributed. Distributed certification finds many applications in the design of fault-tolerant distributed algorithms. We believe that streaming certification can be useful for problems that are hard in streaming, for which computing a solution must be deported to external powerful computing entities (e.g., cloud computing). This external entity may be not trustable or subject to errors. Yet, it can be asked to provide not only the solution, but also a certificate for that solution, such that the solution can then be checked by a semi-streaming algorithm.   

An intriguing direction for future research is the exploration of randomness, aiming to determine to what extent randomization can reduce the space complexity of streaming certification schemes with bounded error.
That is, the study of MA-type of protocols, instead of just NP-type protocols. 
In fact, one could envision considering interactive protocols, e.g., AM (Arthur-Merlin) or even MAM (Merlin-Arthur-Merlin) with multiple passes, etc., in the context of streaming algorithms, as it was successfully recently achieved in the context of distributed computing. 

%%%%%%%%%%%%%%%%%%%%%%%%%%%%%%%%%%%%%%%%%%%%%%%%%%%%%%%%%%%%%%%%
\bibliographystyle{abbrv}
\bibliography{nondet-streaming-bib}

%%%%%%%%%%%%%%%%%%%%%%%%%%%%%%%%%%%%%%%%%%%%%%%%%%%%%%%%%%%%%%%%%
\end{document}